%% file: conference_101719.tex
\documentclass[10pt,conference]{IEEEtran}
\IEEEoverridecommandlockouts
\usepackage{cite}
\usepackage{amsmath,amssymb,amsfonts}
\usepackage{algorithmic}
\usepackage{graphicx}
\usepackage{textcomp}
\usepackage{xcolor}
\def\BibTeX{{\rm B\kern-.05em{\sc i\kern-.025em b}\kern-.08em
    T\kern-.1667em\lower.7ex\hbox{E}\kern-.125emX}}

\usepackage{xspace}

\usepackage{multicol}
\usepackage{multirow}
\usepackage{bigstrut}
\usepackage{xspace}
\usepackage{rotating}
\usepackage[linesnumbered,boxed,ruled,comment snumbered,vlined]{algorithm2e}
\usepackage{amsthm}
\usepackage{float}

\newcommand{\LOGION}{\texttt{LOGION}\xspace}

\newcommand{\myrename}{\textsf{Rename}}
\newcommand{\myinsert}{\textsf{Insert}}
\newcommand{\mydelete}{\textsf{Delete}}

\renewcommand{\phi}{\varphi}

\newcommand{\pltl}{LTL~}
\newcommand{\LtlP}{\mathbb{P}}
\newcommand{\Ltlmodels}{\models}
\newcommand{\LtlM}{M}
\newcommand{\X}{\bigcirc}
\newcommand{\nexttime}{\X}
\newcommand{\futrue}{\Diamond}
\newcommand{\future}{\Diamond}
\newcommand{\globally}{\square}

\newtheorem{definition}{Definition}
\newtheorem{example}{Example}
\newtheorem{property}{Property}
\newtheorem{theorem}{Theorem}

\newcommand{\li}{\mathrm{li}}
\newcommand{\nt}{\mathrm{nt}}

\newcommand{\until}{\mathcal{U}}
\newcommand{\release}{\mathcal{R}}
\newcommand{\weakuntil}{\mathcal{W}}

\newcommand{\ie}{{\it i.e.}}

\newcommand{\wrt}{{\it w.r.t.}~}

\begin{document}

\title{
    Structural Similarity of Boundary Conditions and an Efficient Local Search Algorithm for Goal Conflict Identification
}


 \author{\IEEEauthorblockN{1\textsuperscript{st} Hongzhen Zhong}
 \IEEEauthorblockA{\textit{School of Data and Computer Science} \\
 \textit{Sun Yat-sen University}\\
 Guangzhou, China \\
 zhonghzh5@mail2.sysu.edu.cn}
 \and
 \IEEEauthorblockN{2\textsuperscript{nd} Hai Wan\IEEEauthorrefmark{2}
 \thanks{\IEEEauthorrefmark{2}Corresponding author}}
 \IEEEauthorblockA{\textit{School of Data and Computer Science} \\
 \textit{Sun Yat-sen University}\\
 Guangzhou, China \\
 wanhai@mail.sysu.edu.cn}
 \and
 \IEEEauthorblockN{3\textsuperscript{rd} Weilin Luo}
 \IEEEauthorblockA{\textit{School of Data and Computer Science} \\
 \textit{Sun Yat-sen University}\\
 Guangzhou, China \\
 luowlin3@mail2.sysu.edu.cn}
 \and
 \IEEEauthorblockN{4\textsuperscript{th} Zhanhao Xiao}
 \IEEEauthorblockA{\textit{School of Data and Computer Science} \\
 \textit{Sun Yat-sen University}\\
 Guangzhou, China \\
 xiaozhh9@mail.sysu.edu.cn}
 \and
 \IEEEauthorblockN{5\textsuperscript{th} Jia Li}
 \IEEEauthorblockA{\textit{School of Data and Computer Science} \\
 \textit{Sun Yat-sen University}\\
 Guangzhou, China \\
 lijia49@mail2.sysu.edu.cn}
 \and
 \IEEEauthorblockN{6\textsuperscript{th} Biqing Fang}
 \IEEEauthorblockA{\textit{School of Data and Computer Science} \\
 \textit{Sun Yat-sen University}\\
 Guangzhou, China \\
 fangbq3@mail3.sysu.edu.cn}
 }

\maketitle
\input{document/abstract.tex}


\input{document/Introduction.tex}

\input{document/Background}

\input{document/Similarity}

\input{document/LOGION}
\input{document/Experiments}
\input{document/Related_Wrok}

\input{document/Conclusion}

\section*{Acknowledgment}
We thank Shaowei Cai for the discussion on the paper; Xiaotong Song for her help on the experiments.
This paper was supported by the Guangdong Province Science and Technology Plan projects (No. 2017B010110011 and  2016B030305007), National Natural Science Foundation of China (No. 61976232; 61573386; 61906216), National Key R\&D Program of China (No. 2018YFC0830600),\\Guangdong Province Natural Science Foundation (No. 2016A030313292, 2017A070706010 (soft science), and 2018A030313086), Guangdong Basic and Applied Basic Research Foundation (No. 2020A1515010642), Guangzhou Science and Technology Project (No. 201804010435), and the Fundamental Research Funds for the Central Universities (19lgpy226).

\bibliographystyle{IEEEtran}
\bibliography{sample-base}

\end{document}

%% file: document/abstract.tex
\begin{abstract}
	In goal-oriented requirements engineering,
	goal conflict identification is of fundamental importance for requirements analysis.
	The task aims to find the feasible situations which make the goals diverge within the domain,
	called \textit{boundary conditions (BCs)}.
	However, the existing approaches for goal conflict identification fail to find sufficient BCs and general BCs which cover more combinations of circumstances.  
	From the BCs found by these existing approaches,
	we have observed an interesting phenomenon that there are some pairs of BCs are similar in formula structure, which occurs frequently in the experimental cases.
	In other words, once a BC is found, a new BC may be discovered quickly by slightly changing the former.
	It inspires us to develop a local search algorithm named \LOGION to find BCs, in which the structural similarity is captured by the neighborhood relation of formulae.
	Based on structural similarity,
	\LOGION can find a lot of BCs in a short time. Moreover, due to the large number of BCs identified, it potentially selects more general BCs from them.
	By taking experiments on a set of cases, we show that \LOGION effectively exploits the structural similarity of BCs.
	We also compare our algorithm against the two state-of-the-art approaches. The experimental results show that \LOGION produces one order of magnitude more BCs than the state-of-the-art approaches and confirm that \LOGION finds out more general BCs thanks to a large number of BCs.
\end{abstract}

%% file: document/Introduction.tex
\section{Introduction}
Requirements engineering is an essential phase of the software development life cycle \cite{chakraborty2012role}.
It is important to attain the correct software requirements specification.
Many researches have demonstrated the significant advantages that formal and goal-oriented approaches bring to the generation of correct specifications \cite{alrajeh2009learning,degiovanni2014automated,ellen2014detecting}.
In such approaches, \textit{domain properties} and \textit{goals}, represented in the \textit{Linear-time Temporal Logic (LTL)}, can capture requirements naturally \cite{van2009requirements}.
\textit{Goal conflict identification} is an important stage of formal and goal-oriented requirements engineering, which aims to find inconsistencies between goals.
The inconsistencies are captured by \textit{boundary conditions (BCs)} under which goals are unsatisfiable as a whole within the domain properties.
Below we illustrate BCs using a MinePump example \cite{kramer1983conic}.

\input{document/example1}

Van Lamsweerde \cite{van2009requirements} pointed out that we must identify as many BCs as possible, which is important to refine requirements.
Identifying more BCs helps us select more general BCs.
Intuitively, the more general BC is of higher quality because it can cover more combinations of circumstances.
Once it is worked out, less general BCs will also be worked out.
In order to find more BCs, 
previous approaches can be mainly categorized into construct‐based approaches and search-based approaches.
Construct‐based approaches focus on generating BCs based on constructing templates \cite{van1998managing} or tableau structure \cite{degiovanni2016goal}.
However, they suffer from scalability issues because the templates and the tableau structure could be generated only in a small requirement specification.
The search-based approach \cite{degiovanni2018genetic} focuses on finding BCs.
It produces lots of LTL formulae and checks whether they are BCs one by one. 
But it searches without guided and wastes a lot of time checking the formulae that are not BCs.

When we fixed our eyes on the BCs computed by the existing approaches to generating BCs \cite{degiovanni2016goal,degiovanni2018genetic},
we found an interesting phenomenon that there exist some pairs of BCs which are similar on the structure.
Recall that $BC_1$ and $BC_2$ are presented in Example \ref{examp:traincrossing}, these two formulae are similar on the structure: they differ from each other in only one place.
Furthermore, $BC_2$ is more general than $BC_1$ because $BC_1$ is just a special case of $BC_2$. Based on structural similarity, we can find more general $BC_2$, not just $BC_1$.
We will give a formal definition of general BC in section \ref{section:Goal-Oriented}.
In consequence, once one BC is found, the other one may be found via simply changing the former. Fortunately, we have observed that this case occurs frequently in the specifications, which allows us to find more BCs quickly and help us to select more general BCs via taking advantage of such a kind of structural similarity.
To formalize the sightly change of formulae, we propose three formula edit operations: \myrename, \myinsert, and \mydelete.
For that, we introduce an efficient local search based algorithm to compute BCs, in which the structural similarity is captured by the neighborhood relation of formulae.
A formula differs from its neighbor on only one or two variables or operators.

In this paper, we evaluate to what extent our approach \LOGION exploits the structural similarity of BCs on a set of classical cases.
The experimental results show that the BCs successively computed by \LOGION are quite similar on the formula structure.

We also compare our approach \LOGION against Tableaux-based  \cite{degiovanni2016goal} and the GA-based \cite{degiovanni2018genetic} approaches.
The experiment results show that our approach \LOGION finds one order of magnitude more BCs than these two approaches in almost all cases.
We also choose more general BCs based on the large number of BCs found by \LOGION.

Our main contributions are summarized as follows.
\begin{enumerate}
  \item From the BCs found by the Tableaux-based \cite{degiovanni2016goal} and the GA-based \cite{degiovanni2018genetic} approaches, we discover an interesting phenomenon that there are some pairs of BCs are similar in formula structure.
  \item We propose an efficient local search algorithm named \LOGION for goal-conflict identification, in which the neighborhood of formula captures the structural similarity of BCs.
  \item Empirical evidence shows that our approach \LOGION has superiority on the efficiency of computing BCs, over the Tableaux-based  \cite{degiovanni2016goal} and the GA-based \cite{degiovanni2018genetic}  approaches.
\end{enumerate}

The remainder of the paper is organized as follows.
In Section \ref{sec:background}, we give preliminaries about goal-conflict identification,
linear-time temporal logic, tree edit distance, and local search algorithms.
Next, we show how the similarity of BCs occurs frequently in Section \ref{sec:similarity}
and describe our approach in detail in Section \ref{sec:LOGION}.
In Section \ref{sec:expriments}, we carry out our
experiments to validate our approach and compare it with related approaches.
Finally, we discuss related work in Section  \ref{sec:relatedwork} and make some conclusions in the last section.

%% file: document/example1.tex
\begin{example}\label{examp:traincrossing}
    Consider a system to control a pump inside a mine. The main goal of the system is to avoid flood in the mine. The system has two sensors. One detects the high water level ($h$), the other detects methane in the environment ($m$). When the water level is high, the system should turn on the pump ($p$). When there is methane in the environment, the pump should be turned off. Domain property and goals are represented via the following LTL formulae.

\noindent \textbf{Domain Property:}
\begin{enumerate}
	\item \textbf{Name}: PumpEffect \\
	\textbf{Description}: The pump is turned on for two time steps, then in the following one the water level is not high. \\
	\textbf{Formula}:
	$\globally((p \wedge \nexttime p) \rightarrow \nexttime(\nexttime \neg h))$
\end{enumerate}
\noindent \textbf{Goals:}
\begin{enumerate}
	\item \textbf{Name}: NoFlooding \\
	\textbf{Description}: When the water level is high, the system should turn on the pump.\\
	\textbf{Formula}: $\globally(h \rightarrow \nexttime(p))$
	\item \textbf{Name}: NoExplosion \\
	\textbf{Description}: When there is methane in the environment, the pump should be turned off.\\
	\textbf{Formula}: $\globally(m \rightarrow \nexttime(\neg p))$
\end{enumerate}
Two of BCs of this specification are:
\begin{align*}
	BC_1\text{= } & \globally(h \wedge m) \\
	BC_2\text{= } & \future(h \wedge m)
\end{align*}
Intuitively, $BC_1$ captures the situation that high water level and methane always happen, while $BC_2$ captures the situation that high water and methane will happen in the future. Under the $BC_1$ or $BC_2$, two goals are unsatisfiable simultaneously within domain property.
\end{example}

%% file: document/Background.tex
\section{Background}\label{sec:background}
In this section, we introduce the background of goal conflict identification, linear-time temporal logic, tree edit distance, and local search algorithm. We briefly recall some basic notions for the rest of the paper.

\subsection{Goal Conflict Identification}\label{section:Goal-Oriented}
In goal-oriented requirements engineering methodologies~\cite{van2009requirements},
\textit{goals} and \textit{domain properties} are formed in linear-time temporal logic.
Goals are prescriptive statements that the system must achieve, while domain properties are descriptive statements that capture the domain about the problem world.
It is unrealistic to require requirements specification to be complete or all goals to be satisfiable ideally, because unanticipated cases may occur.
It suggests identifying the conflicts as early as possible.
The goal-conflict identification stage is  important in the conflict analysis phase.

The \textit{conflict analysis phase}~\cite{van1998integrating,van2009requirements} has three main stages:
\begin{enumerate}
	\item the \textit{identification stage} is to identify condition whose occurrence makes the goals diverge;
	\item the \textit{assessment stage} is to assess and prioritize the identified conflicts according to their likelihood and severity;
	\item the \textit{resolution stage} is to resolve the identified conflicts by providing appropriate countermeasures.
\end{enumerate}


In this paper, we focus on the identification stage. 
A conflict represents a condition that occurrence results in the loss of satisfaction of the goals.
\begin{definition}
	Given a set $\{G_1, \dots, G_n\}$ of goals and a set $Dom$ of domain properties, the goals are said to be divergent in the context of $\varphi$ if there exists an expression $\varphi$, called a \emph{boundary condition}, such that the following properties hold:
	\begin{flalign*}
	&Dom \wedge G  \wedge \varphi \models \bot    \tag{\text{logical inconsistency}} \\
	&Dom \wedge  G_{-i} \wedge \varphi \not \models \bot \text{, for  each } \! 1 \leq i \! \leq \! n  \tag{\text{minimality}} \\
	& \neg G \not \equiv \varphi \tag{\text{non-triviality}}
	\end{flalign*}
	\emph{where} $G=\bigwedge_{1\leq i \leq n}G_i$ \emph{and} $G_{-i} = \bigwedge_{j\not =i}G_j$.
\end{definition}

Intuitively, a BC captures a situation where the goals as a whole are not satisfiable.
The logical inconsistency condition means that the conjunction of goals $G_1, \dots, G_n$ becomes inconsistent when $BC$ holds.
The minimality condition state that disregarding one of the goals no longer results in the consistency.
The non-triviality condition forbids a BC to be a trivial condition, which is the negation of the conjunction of the goals.
Note that the minimality condition requires a BC itself to be consistent.

Specifying software requirements in the LTL formulation allows us to employ automated LTL satisfiability (SAT) solvers to check for the feasibility of the corresponding requirements.
With an efficient LTL SAT solvers, we can automatically check if the generated candidate formulae are valid BCs or not by checking if they satisfy the properties.

Given a set of BC $S$, we call BC $\varphi_i \in S$ is more general than another BC $\varphi_j \in S$ if $\varphi_j$ implies $\varphi_i$.
Intuitively, a more general BC $\varphi$ captures all the particular combinations of circumstances captured by the less general BCs than $\varphi$. Therefore, it is also important to find more general BCs. 

\subsection{Linear-Time Temporal Logic}
Linear-Time Temporal Logic (LTL) \cite{emerson1990temporal} is widely used to describe
infinite behaviors of discrete systems.
\pltl formulae are defined from a countably infinite set $\LtlP$ of propositional variables, classical propositional connectives including $\neg$, $\wedge$ , $\vee$ 
and temporal operators  $\globally$ (always), $\futrue$ (eventually), $\X$ (next), $\until$ (until), $\release$ (release) and $\weakuntil$ (weak-until), as follows:
\begin{enumerate}
    \item $b\in \mathbb{B}$ is an LTL formula, where $\mathbb{B} = \{\top, \bot\}$;
    \item every proposition $p\in \LtlP$ is an LTL formula;
    \item if $\varphi_1$ and $\varphi_2$ are LTL formulae, then so are
    $\neg \varphi_1$, $\varphi_1 \wedge \varphi_2$,
    $\varphi_1 \vee \varphi_2$,
    $\X \varphi_1, \globally\varphi_1, \futrue\varphi_1,$
    $\varphi_1 \until \varphi_2,\varphi_1 \release \varphi_2,\varphi_1 \weakuntil \varphi_2$.
\end{enumerate}

LTL formulae are interpreted over linear-time structures.
A linear-time structure is a pair of $\LtlM= (S,\varepsilon)$ where $S$ is an $\omega$-sequence $s_0, s_1,...$ of states
and $\varepsilon: S \rightarrow 2^\LtlP$ is a function mapping each state $s_i$ to a set of propositional variables that hold in $s_i$.
Let $\LtlM$ be a linear-time structure, $i \in \mathbb{N}^0$ a position, and $\phi,\psi$ \pltl formulae.
We define the satisfaction relation $\Ltlmodels$ as follows:
\begin{center}\begin{tabular}{lcl}
$\LtlM,i \Ltlmodels p$ &iff& $p \in \varepsilon(s_i) \text{, where } p \in \LtlP $
\\
$\LtlM,i \Ltlmodels \neg \phi$ &iff& $\LtlM,i \not \Ltlmodels \phi$
\\
$\LtlM,i \Ltlmodels \phi \wedge \psi$ &iff& $\LtlM,i \Ltlmodels \phi  \textsf{ and } \LtlM,i \Ltlmodels \psi$
\\
$\LtlM,i \Ltlmodels \nexttime \phi$ &iff& $\LtlM,i+1 \Ltlmodels \phi$
\\
$\LtlM,i \Ltlmodels \phi ~\until~ \psi$ &iff& $\exists k \geq i \text{ s.t. } \LtlM,k \Ltlmodels \psi \text{ and }$\\
&& $\forall i \leq j {<} k, \LtlM,j \Ltlmodels \phi$
\end{tabular}\end{center}

Operator release ($\release$), eventually ($\future$), always ($\globally$), and weak-until ($\weakuntil$) are commonly used, and can be defined as 
$\varphi_1 \release \varphi_2 := \neg(\neg\varphi_1 \until \neg \varphi_2)$, 
$\future \varphi := \top \until \varphi$, 
$\globally \varphi := \neg (\top \until \neg \varphi)$, 
and $\varphi_1 \weakuntil \varphi_2 := \varphi_1 \until (\varphi_2 \vee \globally \varphi_1)$, respectively.

An \pltl formula $\phi$ is satisfiable if there exists a linear-time structure $\LtlM$ such that $\LtlM, 0 \Ltlmodels \phi$, and it is valid if $\LtlM, 0 \Ltlmodels \phi$ for all linear-time structures $\LtlM$.
The LTL satisfiability and validity problems are decidable and PSPACE-complete~\cite{sistla1985complexity}.

An LTL formula $\varphi$ implies an LTL formula $\varphi'$, in notation $\varphi \models \varphi'$, when $\LtlM, 0 \models \varphi$ for every linear-time structure $\LtlM$ such that $\LtlM, 0 \models \varphi'$. Two LTL formulae $\varphi_1$ and $\varphi_2$ are said to be logically  equivalent, denoted by $\varphi_1 \equiv \varphi_2$, if $\varphi_1 \models \varphi_2$ and $\varphi_2 \models \varphi_1$.

An LTL formula $\psi$ is a subformula of an LTL formula $\varphi$ if $\psi$ is a part of $\varphi$.
We use $|\phi|$ to denote the size of the formula $\phi$, \emph{i.e.}, the number of variables and operators in $\phi$.

We refer the reader to~\cite{manna2012temporal} for further details on LTL.

\subsection{Tree Edit Distance}
The tree edit distance~\cite{zhang1989simple} is used to measure the similarity between two ordered labeled trees and has successfully been applied in a wide range of applications.
The \textit{tree edit distance} is the cost of the minimal-cost sequence of node edit operations that transforms one tree into another. For ordered label trees, there are three node edit operations:
\begin{itemize}
    \item \textit{relabel} the label of a node in tree;
    \item \textit{insert} a node between an existing node and a subsequence of consecutive children of this node;
    \item \textit{delete} a non-root node and connect its children to its parent maintaining the order.
\end{itemize}
Generally, the cost of each node edit operation is one, and the cost of a sequence is the sum of the cost of its node edit operations.
So, the tree edit distance is considered as the length of the sequence with the minimal cost.
We use $|T|$ to denote the size of the tree $T$, \emph{i.e.}, the number of nodes in $T$.
Given two trees $T_1$ and $T_2$, the tree edit distance is denoted by $\delta(T_1, T_2)$.

In fact, for the trees with large sizes,
the tree edit distance becomes unsuitable to represent the difference between these trees.
For example, a tree edit distance of $5$ means a big gap between two trees with the size $10$ but a small gap between two trees with the size $1000$.
To capture the relative similarity \wrt the size, the notion of the normalized tree edit distance was proposed in \cite{rico2003comparison}, which is defined as:
$$\Delta(T, T') = \frac{\delta(T,T')}{|T| + |T'|} \in [0, 1].$$

For further details on the tree edit distance, we refer the reader to~\cite{bille2005survey}.

\subsection{Local Search}
Local search is a kind of meta-heuristic search algorithm.
More formally, given a problem instance $ \pi $, we use $S(\pi)$ to denote its search space, which is the set of all candidate solutions. The {neighborhood function} $ N: S(\pi) \mapsto 2^{S(\pi)} $ maps each candidate solution to its neighbors. The set $N(s)$ is called the neighbors of $s$. The {objective function} $ f:S(\pi) \mapsto \mathbb{R} $ is a mapping of candidate solutions to their objective function value.
Typically, a local search algorithm constructs an initial candidate solution and modifies it iteratively.
At each search step, the algorithm evaluates the neighbors of the current candidate solution by the objective function and move to the best neighbor.
The search procedure terminates when it runs out of time or the best solution found has not been improved in a given number of steps.

Due to the limited amount of local information when choosing a neighbor to move to, the local search suffers from {cycling}, that is, some candidate solutions of high quality are being frequently revisited.
\textit{Tabusearch}~\cite{glover1986future,hansen1990algorithms} is a fundamentally different approach to reduce cycling which forbids steps to recently visited candidate solutions. The simplest and most widely applied
implementation of tabu search consists of an iterative improvement algorithm enhanced with a form of short-term memory $ M $ which stores the last $ T $ visited solutions, where $ T $ is called the \textit{tabu tenure}. In each search step, the algorithm chooses the best neighbor in $ N(s) \backslash M $, where $ s $ is the current candidate solution.

For further details on the local search algorithm, we refer the reader to \cite{hoos2004stochastic}.

%% file: document/Similarity.tex
\section{Similarity of Boundary Conditions}\label{sec:similarity}
As we show in the introduction section,
we have observed that there are some pairs of BCs similar on the structure.
In this section, we show the structural similarity of BCs in the classical cases.

\subsection{Cases}

We use 16 cases introduced by Degiovanni \cite{degiovanni2018genetic}.
Table \ref{tab:casestudy} summarizes the numbers of domain properties (``\#Dom''), goals (``\#Goal''), variables (``\#Var''), and the total size of all formulae (``Size'') for the specification of each case study.
\input{document/table1}

\subsection{Structural Similarity}
In order to capture the structural similarity of LTL formulae,
we borrow the notion of the similarity on ordered label trees.
It is notable that every LTL formula corresponds to a parse tree which is also an ordered label tree \cite{enderton2001mathematical}.
Then we use $T_\phi$ to denote the parse tree of the formula $\phi$.
Note that the size of the formula $\phi$ equals to the size of its parse tree $T_\phi$, \emph{i.e.}, $|T_\phi|=|\phi|$.
The parse tree of the formula $\globally(h\rightarrow \nexttime (p))$ is shown in Figure \ref{fig:formulaEditingOperations}(a).

Next, we define the distance between two formulae as the tree edit distance between their corresponding parse trees.
Formally, for two formulae $\phi$ and $\phi'$,
we use $\delta(\phi,\phi') = \delta(T_{\phi},T_{\phi'})$ to denote the formula distance between $\phi$ and $\phi'$.
As BCs in different specifications may differ seriously on the formula size,
we also consider the relative distance between two formulae \emph{w.r.t.} their sizes.
Similarly, we use $\Delta(\phi,\phi') = \Delta(T_{\phi},T_{\phi'})$ to denote the normalized formula distance between $\phi$ and $\phi'$.

Intuitively, the distance between two formulae indicates their divergence on the structure:
the distance is bigger, the divergence is bigger.
Indeed, the formula distance also has the property of the tree edit distance:
\begin{itemize}
  \item $\delta(\phi,\phi) = \Delta(\phi,\phi) = 0$
  \item $\Delta(\phi,\phi') \in [0,1]$
\end{itemize}

\begin{example}[Example \ref{examp:traincrossing} cont.]
	Let $\varphi_1 = \globally(h \rightarrow \nexttime(p))$, $\varphi_2 = \globally(h \wedge m)$ and $\varphi_3 = \future(h \wedge m)$.
	For formula distance, $\delta(\varphi_1, \varphi_2)=3$ and $\delta(\varphi_2, \varphi_3) = 1$.
	For normalized formula distance, $\Delta(\varphi_1, \varphi_2)=0.333$ and
	$\Delta(\varphi_2, \varphi_3) = 0.125$.
\end{example}

According to the similarity notion of LTL formulae,
here we have our first research question:

\textbf{RQ1:} How frequently similar BC pairs do occur in these cases?

As there is no set of BCs for these specifications, we generated BCs by applying the existing approaches \cite{degiovanni2016goal, degiovanni2018genetic}
for 24 hours (10 runs in parallel).
After removing the identical BCs, the numbers of BCs of different cases are shown in Table \ref{tab:rq1-1}, denoted by ``\#total BC''.

In order to explore how frequently similar BC pairs occur, for each case study, we compute
the proportion of BCs which have at least a similar BC whose formula distance is not greater than $l$ $(l = 1,2,3)$, on the total number of BCs.
Formally, we use $\%BC(\delta \leq l)$ to denote the proportion, which is defined as:
$$\%BC(\delta \leq l) = \frac{\#\{BC | \exists BC'. \delta(BC,BC')\leq l \}}{\text{\#total BC}}.$$

We also compute the average number of similar BCs of each BC whose distance is not greater than $l$.
Formally, given a set of BCs and a boundary condition $BC$, we use $sim(BC,l)=\{BC' | \delta(BC,BC')\leq l \}$ to denote the set of BCs whose formula distance to $BC$ is not greater than $l$.
We use $\#sim(\delta \leq l)$ to denote the average number of similar BCs, which is defined as:
$$\#sim(\delta \leq l) = \frac{\sum\#sim(BC,l)}{\text{\#total BC}}.$$

The results are shown in Table \ref{tab:rq1-1}.
Taking the example of LAS,
for each boundary condition $BC$, there are averagely $6$ BCs which differ from $BC$ only on one or two variables or operators.
From the table, we can observe that in most cases except for Load Balancer and LiftController,
there are more than $70\%$ BCs which can be obtained by applying two or three formula edit operations from another BC.
In these cases, every BC averagely has more than one BCs with at most variables or operators differing.

Besides the formula distance, we also consider the normalized formula distance in the cases of different upper bound $k$ $(k = 0.1,\ 0.2,\ 0.3)$.
The results are also shown in Table \ref{tab:rq1-1}.
For example, in RRCS, every boundary condition $BC$ averagely has $6.6$ BCs which has a normalized formula distance to $BC$ less than $0.2$.
The results illustrate that
in each case, there are more than $90\%$ BCs having a relatively similar BC whose normalized formula distance is at most $0.3$.
The results also show that
each BC averagely has more than two similar BC whose normalized formula distance is at most $0.2$.

Next, in order to explore how close the similar BC pairs are, for each case study, we compute the average minimum distance between BCs and their most similar BCs.
The results are shown in the column ``avg min'' in Table \ref{tab:rq1-1}.
It shows that in most cases, most BCs have a similar BC whose formula distance is less than $2$.

To sum up, 
in most cases, there are more than $70\%$ BCs which have both an absolutely similar BC $(\delta \leq 3)$
a relatively similar BC $(\Delta \leq 0.3)$. Similar BC pairs are very common in cases.

\input{document/table2}

%% file: document/table1.tex
\begin{table}[t]
	\centering
	\caption{The details of each case include the numbers of domain properties (\#Dom), goals (\#Goal), variables (\#Var), and the total size of all formulae (Size) for the specification of each case}
\begin{tabular}{c|c|c|c|c}
	\hline
	Case  & \#Dom & \#Goal & \#Var & Size \bigstrut\\
	\hline
	RetractionPattern1 (RP1) & 0     & 2     & 2     & 9 \bigstrut[t]\\
	RetractionPattern2 (RP2) & 0     & 2     & 4     & 10 \\
	Elevator (Ele) & 1     & 1     & 3     & 10 \\
	TCP   & 0     & 2     & 3     & 14 \\
	AchieveAvoidPattern (AAP) & 1     & 2     & 4     & 15 \\
	MinePump (MP) & 1     & 2     & 3     & 21 \\
	ATM   & 1     & 2     & 3     & 22 \\
	RRCS  & 2     & 2     & 5     & 22 \\
	Telephone (Tel) & 3     & 2     & 4     & 31 \\
	LAS   & 0     & 5     & 7     & 32 \\
	Prioritized Arbiter (PA) & 6     & 1     & 6     & 57 \\
	Round Robin Arbiter (RRA) & 6     & 3     & 4     & 77 \\
	Simple Arbiter (SA) & 5     & 3     & 6     & 84 \\
	Load Balancer (LB) & 3     & 7     & 5     & 85 \\
	LiftController (LC) & 7     & 8     & 6     & 124 \\
	AMBA  & 6     & 21    & 16    & 415 \bigstrut[b]\\
	\hline
\end{tabular}%
	\label{tab:casestudy}
\end{table}%

%% file: document/table2.tex
\begin{table*}[ht]
    \centering
    \caption{Experimental results on the similarity of BCs computed by the GA-based approach and Tableaux-based approach (Tableaux for short)}
    \label{tab:rq1-1}%
\begin{tabular}{c|ccc|ccc|ccc|ccc|c|c}
	\hline
	\multirow{3}[2]{*}{Case} & \multicolumn{14}{c}{GA-based Approach and Tableaux} \bigstrut[t]\\
	& \multicolumn{3}{c|}{$\%BC(\delta \leq l)$} & \multicolumn{3}{c|}{$\#\mathsf{sim}(\delta \leq l)$} & \multicolumn{3}{c|}{$\%BC(\Delta \leq k)$} & \multicolumn{3}{c|}{$\#\mathsf{sim}(\Delta \leq k)$} & \multirow{2}[1]{*}{avg min} & \multirow{2}[1]{*}{\#total BC} \\
	& $l=1$ & $l=2$ & $l=3$ & $l=1$ & $l=2$ & $l=3$ & $k=0.1$ & $k=0.2$ & $k=0.3$ & \!\!\!$k=0.1$ \!\!\!&\!\!\! $k=0.2$ \!\!\!& \!\!\!$k=0.3$\!\!\! &       &  \bigstrut[b]\\
	\hline
	RP1   & 85.3\% & 96.4\% & 98.0\% & 2.1   & 6.8   & 13.4  & 91.9\% & 97.1\% & 98.7\% & 5.1   & 21.0  & 42.3  & 1.1   & 308  \bigstrut[t]\\
	RP2   & 84.2\% & 96.6\% & 98.8\% & 2.1   & 6.8   & 13.5  & 87.6\% & 97.2\% & 99.1\% & 8.1   & 39.4  & 61.0  & 1.1   & 566  \\
	Ele   & 65.9\% & 93.5\% & 97.6\% & 1.2   & 4.2   & 9.3   & 35.8\% & 90.2\% & 97.6\% & 1.0   & 4.8   & 11.6  & 1.2   & 124  \\
	TCP   & 69.3\% & 91.9\% & 98.7\% & 1.5   & 5.0   & 11.3  & 52.8\% & 93.2\% & 98.2\% & 1.7   & 6.9   & 16.2  & 1.1   & 382  \\
	AAP   & 63.2\% & 92.0\% & 97.9\% & 1.4   & 5.0   & 11.7  & 24.7\% & 85.8\% & 95.1\% & 0.5   & 3.4   & 7.3   & 1.2   & 289  \\
	MP    & 68.4\% & 93.4\% & 98.2\% & 1.5   & 5.4   & 12.7  & 54.2\% & 90.8\% & 96.8\% & 1.9   & 10.0  & 22.0  & 1.1   & 913  \\
	ATM   & 72.6\% & 94.6\% & 98.5\% & 1.7   & 5.8   & 12.0  & 77.9\% & 97.2\% & 98.5\% & 3.6   & 17.5  & 37.2  & 1.1   & 681  \\
	RRCS  & 40.8\% & 81.0\% & 90.2\% & 0.7   & 2.3   & 4.3   & 66.5\% & 93.4\% & 98.1\% & 1.6   & 6.6   & 14.7  & 1.5   & 317  \\
	Tel   & 39.2\% & 72.0\% & 82.9\% & 0.7   & 1.9   & 3.7   & 80.5\% & 95.1\% & 97.9\% & 3.1   & 10.5  & 23.9  & 1.7   & 534  \\
	LAS   & 80.4\% & 92.5\% & 99.1\% & 2.0   & 6.0   & 8.9   & 95.3\% & 97.2\% & 97.2\% & 18.0  & 24.4  & 27.4  & 1.1   & 108  \\
	PA    & 58.8\% & 79.1\% & 87.8\% & 1.1   & 3.0   & 5.6   & 74.7\% & 86.9\% & 92.8\% & 2.5   & 7.3   & 15.5  & 1.5   & 321  \\
	RRA   & 86.2\% & 95.6\% & 97.8\% & 2.1   & 7.0   & 15.1  & 92.5\% & 98.1\% & 99.4\% & 4.0   & 16.8  & 36.6  & 1.1   & 319  \\
	SA    & 29.4\% & 64.7\% & 76.5\% & 0.5   & 1.1   & 1.5   & 64.7\% & 82.4\% & 88.2\% & 2.1   & 3.8   & 7.1   & 1.8   & 18  \\
	LB    & 7.1\% & 29.8\% & 41.7\% & 0.2   & 0.5   & 0.7   & 63.1\% & 84.5\% & 91.7\% & 0.8   & 2.1   & 3.1   & 3.6   & 85  \\
	LC    & 0.9\% & 12.4\% & 27.8\% & 0.1   & 0.3   & 0.6   & 36.9\% & 66.5\% & 81.7\% & 0.9   & 3.4   & 6.8   & 4.2   & 443  \\
	AMBA  & 1.8\% & 43.3\% & 72.0\% & 0.1   & 1.1   & 2.0   & 38.4\% & 76.8\% & 91.5\% & 1.1   & 6.9   & 12.5  & 1.9   & 165  \bigstrut[b]\\
	\hline
\end{tabular}%
    \normalsize
    \begin{flushleft}%
        ``$\%BC(\delta \leq l)$'' and ``$\%BC(\Delta \leq k)$'' mean the proportion of BCs which have at least a similar BC whose formula distance is at most $l$ and whose normalized formula distance is at most $k$, respectively.
       ``$\#\mathsf{sim}(\delta \leq l)$'' and ``$\#\mathsf{sim}(\Delta \leq l)$'' mean the average number of similar BCs of each BC whose formula distance are at most $l$ and whose normalized formula distance is at most $k$.
        ``$avg\ min$'' means the average distance between BCs and their most similar ones.
        ``\#total BC'' stands for the total number of different BCs found by the GA-based approach and Tableaux-based approach for 24 hours (10 runs in parallel).
    \end{flushleft}
\end{table*}

%% file: document/LOGION.tex
\section{The \LOGION Algorithm}\label{sec:LOGION}
In this section, we develop a local search algorithm named \LOGION for goal conflict identification,
which exploits the structural similarity of BCs to design the neighborhood relation.
Next, we specify each component of our algorithm \LOGION and then give its entire description.

\subsection{Neighborhood}
The search space of \LOGION is composed of LTL formulae.
To search for BCs, we define the syntactic neighborhood of an LTL formula.
Inspire by node edit operations, we introduce LTL \textit{formula edit operations}, which can slightly modify an LTL formula to another LTL formula.

Let $o_1, o_1'$ be two different LTL unary operators (\ie, $o_1, o_1' \in \{\neg,\globally,\futrue,\X\}$).
and $o_2, o_2'$ be two different LTL binary operators (\ie,
$o_2, o_2' \in \{\wedge,\vee, \until,\release,\weakuntil \}$).
Given a formula $\psi$, its neighbor $\psi'$ can be constructed by one of three \textit{formula edit operations} defined below:
\begin{enumerate}
    \item \myrename \label{em:rename}
    \begin{enumerate}
        \item $\psi' = p'$, when $\psi=p$, $p \in \LtlP \cup \mathbb{B}$, $p' \in \LtlP \cup \mathbb{B}$ and $p \not = p'$ \label{em:rename2}
        
        \item $\psi' = o_1'\psi_{1}$, when $\psi = o_1\psi_{1}$ \label{em:rename3}
        \item $\psi' = \psi_{1}o_2'\psi_{2}$, when $\psi = \psi_{1}o_2\psi_{2}$ \label{em:rename4}

    \end{enumerate}
    \item \myinsert \label{em:insertion}
    \begin{enumerate}
        \item $\psi' = o_1'\psi$ \label{em:insertion1}
        \item $\psi' = \psi o_2'p$ or $p o_2'\psi$, when $p \in \LtlP \cup \mathbb{B}$ \label{em:insertion2}
    \end{enumerate}
    \item \mydelete \label{em:deletion}
    \begin{enumerate}
        \item $\psi' = \psi_{1}$, when $\psi = o_1\psi_{1}$ \label{em:deletion1}
        \item $\psi' = \psi_{1}$ or $\psi' = \psi_{2}$, when $\psi = \psi_{1}~o_2~\psi_{2}$\label{em:deletion2}
    \end{enumerate}
\end{enumerate}
The \myrename~operation replaces the top symbol of $\psi$ with other same type symbol.
The \myinsert~adds a new LTL operator to be top operator of $\psi$. Note that if the new LTL operator is binary operator, it need to add a new proposition $p$ for creating the valid formula $\psi'$.
The \mydelete~removes the top operator of $\psi$. If the top operator of $\psi$ is a binary operator. one of the children is a proposition. Then, the \mydelete~removes the top operator and the proposition.

\begin{theorem}
	For any formulae $\psi_1$ and $\psi_2$, $\psi_1$ can be changed to $\psi_2$ by a limited number of formula edit operations.
\end{theorem}
\begin{proof}
	We can change $\psi_1$ to a proposition by \mydelete~and \myrename. 
	Then we change the proposition to $\psi_2$ by \myinsert~and \myrename.
\end{proof}

Based on the formula edit operation, we design the neighborhood of an LTL formula.
\begin{definition}
	Given an LTL formula $\varphi$ as input, the neighborhood function w.r.t. $\varphi$, denoted by $N(\varphi)$, returns all formulae obtained
	by applying a formula editing operation to a sub-formulae of $\varphi$.
\end{definition}

Specifically, we first extract all subformulae of $\varphi$, then we modify each subformula with the above formula edit operations to get $N(\varphi)$.

\begin{property}\label{proerty:distance}
	For an LTL formula $\varphi$, if $\varphi' \in N(\varphi)$, then $\delta(\varphi', \varphi)=1$ or $2$.
\end{property}
\begin{proof}
	As $\phi'$ is a neighbor of $\phi$, $\varphi'$ is obtained from $\phi$ via a formula edit operation.
	Then we consider the formula edit operation in three cases:
	\begin{enumerate}
		\item In the case of \myrename~operation, \emph{i.e.}, (\ref{em:rename2})-(\ref{em:rename4}),  $T_{\varphi'}$ is obtained from $T_\varphi$ by relabeling the label of the corresponding node. So $\delta(\varphi', \varphi)=1$.
		\item In the case of unary \myinsert~operation (\ref{em:insertion1}) or unary \mydelete~operation (\ref{em:deletion1}),  $T_{\varphi'}$ is obtained by inserted/deleted the corresponding node. So $\delta(\varphi', \varphi)=1$.
		\item In the case of binary \myinsert~operation (\ref{em:insertion2}) or binary \mydelete~operation (\ref{em:deletion2}), $T_{\varphi'}$ is obtained from $T_\varphi$ by inserting or deleting two corresponding nodes. So $\delta(\varphi', \varphi)=2$.
	\end{enumerate}
\end{proof}
Property \ref{proerty:distance} shows that the structural similarity of formulae, as all the neighbors of a formula are formulae whose distance to it is at most $2$.

Since there are too many neighbors for $\varphi$ 
(such as the $\globally(h \wedge m)$ of Example \ref{examp:traincrossing}, the number of all its neighbors is $90$), to efficiently select a good candidate BC, we bound the number of neighbors to $k$ in $N(\varphi)$, denoted by $N_k(\varphi)$, where $k$ is a hyperparameter. 
Specifically, we first select $k$ sub-formulae of $\varphi$ uniformly at random. Then, for
each sub-formulae, select a formula edit operation uniformly at
random to modify it.

\begin{example}[Example \ref{examp:traincrossing} cont.]
\label{exp:formulaEditOpt}
Considering the formula $\globally(h\rightarrow \nexttime (p))$ and $k=3$, suppose we select the subformula $\psi$ of $\varphi$, \ie, $h$, $h \rightarrow \nexttime(p)$ and $\nexttime(p)$. Figure \ref{fig:formulaEditingOperations}(b), (c) and (d) show the result of \myrename~, \myinsert~and \mydelete~, respectively.
\end{example}

\begin{figure}[t]
	\centering
	\includegraphics[width=0.45\textwidth]{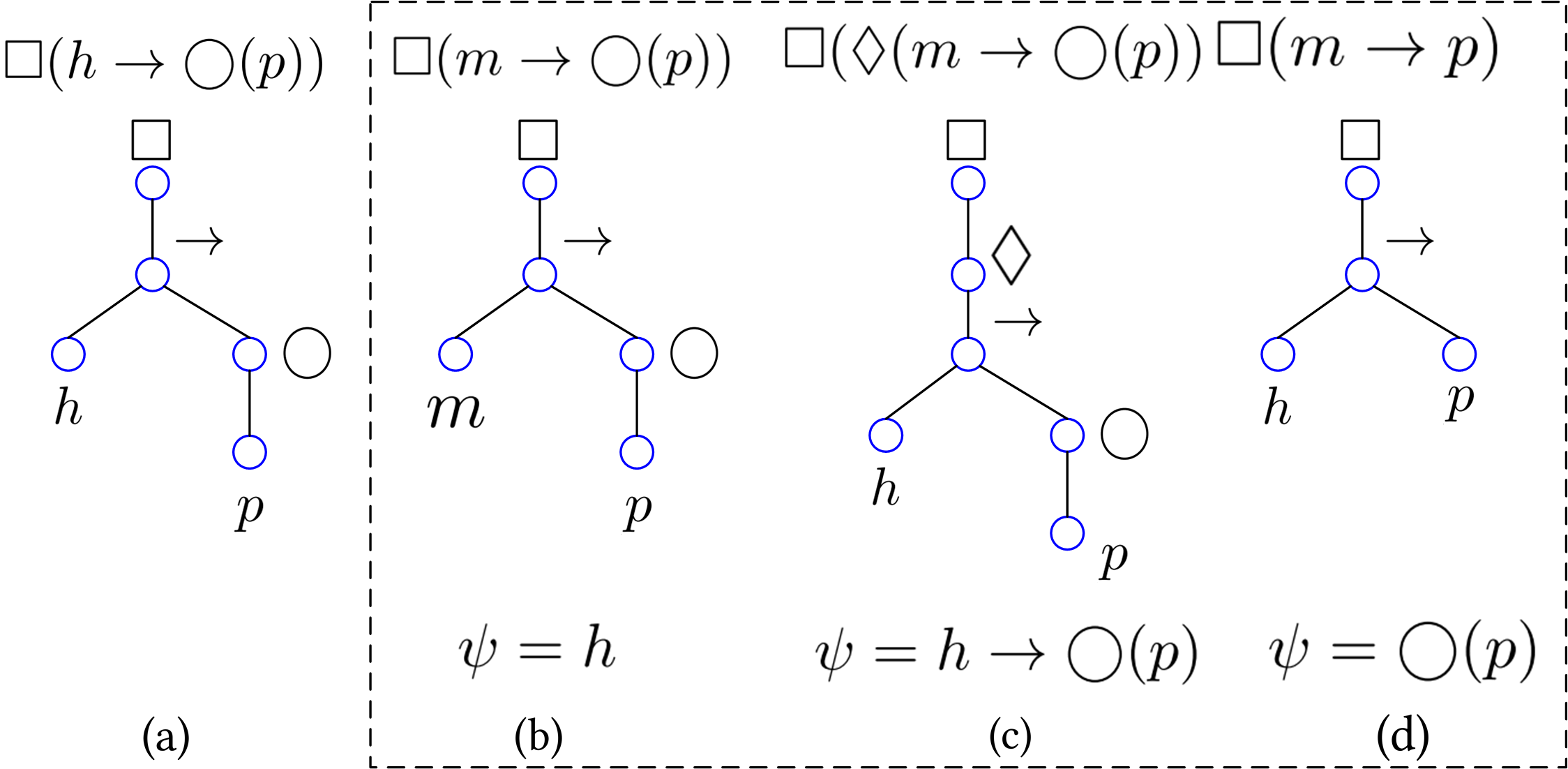}
	\caption{An example of getting neighbors of $\globally(h\rightarrow \nexttime (p))$ by formula editing operations. (a) $\globally(h \rightarrow \nexttime (p))$ and its parse tree. (b) The neighbor formula after renaming $h$ to $m$ and its parse tree. (c) The neighbor formula after inserting $\futrue$ and its parse tree. (d) The neighbor formula after deleting $\nexttime$ and its parse tree.}
	\label{fig:formulaEditingOperations}
\end{figure}

In order to avoid that some formulae are frequently visited, we employ the tabu search strategy~\cite{glover1986future,hansen1990algorithms}.
An LTL formula is called a $T$-tabu formula iff it has been visited during the last $T$ search steps, where $T$ is called tabu tenure and is a hyperparameter.
For an LTL formula $\varphi$, its neighborhood \wrt a $T$-tabu is defined as its neighbors which are not $T$-tabu formulae. formally,
$N_k(\varphi, T) = \{\beta \mid \beta \in N_k(\varphi)$ and $\beta$ is not a $T$-tabu formula$\}$.

\subsection{Initialization and Objective Function}\label{sec:initializationAndObjective}
In order to find BCs quickly, the initial formula should be as ``close'' to a BC as possible.
Therefore, we propose $\neg (G_1 \wedge \cdots \wedge G_n)$, called \textit{trivial condition}, as the initial formula.
It is easily constructed and ``close'' to a BC because it only violates non-triviality in Section \ref{sec:background}.

\begin{example}[Example \ref{examp:traincrossing} cont.]
	Consider the Example \ref{examp:traincrossing}, the initial formula is
	$\neg(\globally(h \rightarrow \nexttime(p)) \wedge \globally(m \rightarrow \nexttime(\neg p))$.
\end{example}

Next, we use an objective function to measure how ``close'' the formula is to a BC.
According to the three properties of BCs,
Degiovanni \emph{et al.}\cite{degiovanni2018genetic} proposed an objective function to capture the similarity between formula and BC.
Here we follow their objective function.
Given an LTL formula $\varphi$, the objective function $f$ is defined as:
$$
f(\varphi)=\li\left(\varphi\right)+\sum_{i=1}^{|G|} \min \left(\varphi, G_{-i}\right)+\nt\left(\varphi\right) + \frac{1}{|\varphi|}
$$
where $|G|$ is the number of the goals and the functions $li$, $min$ and $nt$ are defined as follows:

\begin{align*}
li(\varphi) &=
\left\{\begin{array}{ll}{1} & {\text { if } (Dom \wedge G \wedge \varphi)   \text{ is UNSAT} } \\
{0} & {\text { otherwise }}
\end{array}\right. \\
min\left(\varphi, G_{-i}\right)&=
\left\{\begin{array}{ll}{
	\frac{1}{|G|}} & {\text { if }(Dom \wedge G_{-i} \wedge \varphi) \text{ is SAT}} \\
{0} & {\text { otherwise }}\end{array}\right. \\
nt\left(\varphi\right)&=\left\{\begin{array}{ll}
{0.5} & {\text { if }   \neg G \not \equiv \varphi} \\
{0} & {\text { otherwise }
}\end{array}\right.
\end{align*}

Intuitively, the first three terms of the objective function $f$ capture the three properties of BC, respectively: \textit{li} for the logical inconsistency, \textit{min} for the minimality and \textit{nt} for the non-triviality.
The last term makes the local search tend to produce compact formulae, which is a secondary issue.
We use an LTL SAT checker, such as Aalta~\cite{li2015sat}, to check the satisfiability of LTL formula in the objective function.

\subsection{The Description of \LOGION}
\begin{algorithm}[t]
    \caption{The {\em \LOGION} Algorithm}
    \label{alg:LOGION}
    \KwIn{the domain properties $Dom$ and a set of goals $G_1, \dots, G_n$}
    \KwOut{a set of boundary conditions $\mathfrak{B}$} 
    $\mathfrak{B} \gets \emptyset$;\\
    $\varphi* \gets \neg(G_1 \wedge \dots \wedge G_n)$;\\ \label{alg:trivial-condition}
    \While{the cutoff time is not reached} {
        \If{$N(\varphi*, T) \not = \emptyset$}{
            $\mathfrak{\bar{B}} \gets$ $N_k(\varphi*, T)$; \\ \label{alg:bms1}
            $\varphi* \gets$ the formula in $\mathfrak{\bar{B}}$ with highest score; \\ \label{alg:bms2}
            $\mathfrak{B} \gets$ $\mathfrak{B}$ $\cup$ BCs in $\mathfrak{\bar{B}}$; \\ \label{alg:addBC}
        }
    	T-tabu update;
    }
    \Return $\mathfrak{B}$; \\ \label{alg:return}
\end{algorithm}

In this section, we give the entire description of the \LOGION algorithm based on the components described above.
\LOGION is shown in Algorithm \ref{alg:LOGION}.
In the beginning, \LOGION uses the trivial condition to construct an initial LTL formula $\varphi*$ (line \ref{alg:trivial-condition}).
After that, \LOGION executes a search step iteratively until the time budget is reached.
During the searching procedure, $\varphi*$ represents the current best formula. $N_k(\varphi*, T)$ is not always empty because T-tabu formulae update at the end of each iteration.
\LOGION computes the score of each formula in $\mathfrak{\bar{B}}$ according to the objective function mentioned in \ref{sec:initializationAndObjective}, and walks to the formula with the highest score (line \ref{alg:bms2}).
\LOGION adds all BCs from $\mathfrak{\bar{B}}$ into $\mathfrak{B}$ (line \ref{alg:addBC}).
Finally, it returns all BCs $\mathfrak{\bar{B}}$ (line \ref{alg:return}).

%% file: document/Experiments.tex
\section{Experiments}\label{sec:expriments}
In this section, we conduct extensive experiments on a broad range of cases shown in Section \ref{sec:similarity} to evaluate the effectiveness of our \LOGION algorithm
by comparing it with the state-of-the-art approaches.
We start by presenting two research questions and presenting the state-of-the-art competitors and experimental preliminaries about the experiments.
Then, we show the experimental results and give some discussions about the empirical results.
Finally, we show an application of BC.

As a start, we propose the following research questions to evaluate our \LOGION algorithm.

\textbf{RQ2:} To what extent does \LOGION exploit the structural similarity of BCs?

\textbf{RQ3:} How does \LOGION compare against the state-of-the-art competitors on cases?

\subsection{The State-of-the-art Competitors}
We compare \LOGION against two state-of-the-art approaches. One~\cite{degiovanni2018genetic} is based on a genetic algorithm, the other~\cite{degiovanni2016goal} is a Tableaux-based LTL satisfiability procedure. we call them as the GA-based approach and the Tableaux-based approach, respectively.

The GA-based approach~\cite{degiovanni2018genetic} is effective in finding multiple BCs in different forms since it applies suitable crossover and mutation operators.

The Tableaux-based approach~\cite{degiovanni2016goal} is a deterministic algorithm. It constructs paths that ``escape'' from the tableau structure and generates a small number of BCs.

\subsection{Experimental Settings}\label{sec:experimetalSettings}

\LOGION is implemented in Java, which uses the Java Metaheuristics Search Framework (JAMES)~\cite{de2017james,de2015james},
and integrating the LTL2B\"{u}chi library~\cite{giannakopoulou2002states} to parse LTL requirements specifications,
and the LTL satisfiability checker Aalta~\cite{li2015sat}.
In our experiments, the tabu tenure $T$ is set to $4$, $k$ is set to $50$.

We run the GA-based approach \footnote{http://dc.exa.unrc.edu.ar/staff/rdegiovanni/ASE2018.html} and \LOGION 10 times on each of the 16 cases.
In addition, we run Tableaux-based approach\footnote{https://dc.exa.unrc.edu.ar/staff/rdegiovanni/ase2016} only one time because it is a deterministic algorithm.
The cutoff time for our approach \LOGION and the GA-based approach is set to 1 hour, while the cutoff time for the Tableaux-based approach is set to 3 hours.

All the experiments were run on the 2.13GHz Intel E7-4830, with 128 GB memory under GNU/Linux (Ubuntu 16.04).

To compare the more general BCs among different algorithms,
we first computed a set of BCs $\Pi$ that fulfills the following three
requirements. (1) $\Pi$ is composed of BC obtained by all algorithms,
(2) $\not \exists \psi, \varphi \in \Pi\ s.t.\ \psi \rightarrow \varphi,$ and (3) for each BC identified by each
algorithm, it is either in $\Pi$ or less general than a BC in $\Pi$. 
Intuitively, $\Pi$ covers all the circumstances captured by all BCs identified by each algorithm. Then, we counted the number of BCs in $\Pi$ identified by each algorithm, denoted by ``\#gen.''. The larger ``\#gen.'' means that the BCs identified by the algorithm as a whole is more general than the ones identified by other algorithms.

In addition,
we use the method in \cite{degiovanni2018goal} to evaluate the likelihood of the BCs. In practice, the likelihood can be used to prioritize conflicts to be resolved, and suggest which goals to drive attention to for refinements. Specifically, we first computed the likelihood of the more general BCs. Then, based on the likelihood, we ranked BCs:
the BC with the larger likelihood ranks higher. We reported the
average ranking of the more general BCs with the highest likelihood of each competitor, denoted by ``rank''. 

\subsection{Experimental Results}
In this subsection, we present the experimental results and then answer the research questions mentioned above.

\textbf{Experiments on to what extent \LOGION exploits the structural similarity of BCs (RQ2):}
For each case study, we compute the average distances ($\delta(BC_{i}, BC_{i+1}) $) and normalized distance ($\Delta(BC_{i}, BC_{i+1}) $) between the average distance between boundary condition $BC_i$ and its next boundary condition $BC_{i+1}$ computed successively by our algorithm \LOGION.
The results are shown in Table \ref{tab:rq1-2}.
From the table, we can observe that in most cases (except LB, LC and AMBA), the $BC_{i}$ and $BC_{i+1}$ successively computed are quite similar. 
In other words, it is possible that once a BC is found, a new BC may be found after two to five neighbor jumps.
For LB, LC and AMBA cases, the normalized distances are less than $0.1$ but it needs more neighbor jumps to find the next BCs because the BCs have a bigger size.

The similarity between the pairs of BCs computed successively reflects that our algorithm \LOGION in fact exploits the structural similarity to find BCs.

\input{document/table3}

\textbf{Experiments on comparing \LOGION against its state-of-the-art competitors (RQ3):}
\input{document/table4}
The results are shown in Table \ref{tab:rq2-1} present that \LOGION performs substantially better on all cases than Tableaux-based approach and GA-based approach.
Firstly, \LOGION can handle more cases than others in one hour.
Secondly, comparing the BCs found (``\#BCs''), \LOGION outperforms the GA-based approach and the Tableaux-based approach by one order of magnitude in almost all cases.
Thirdly, \LOGION generally have higher ``\#gen.'' and ``rank'' than others.
It indicates that the local search algorithm exhibits a faster search process to search for BCs than the tableaux and genetic algorithm. \LOGION potentially identifies more general BCs based on a large among of BCs.
This shows that \LOGION not only finds more BCs but also more general BCs.
Fourthly, for the first 15 cases, we can observe that the success rate of \LOGION is higher than that of the two approaches, which demonstrates that our approach \LOGION is more robust.

Besides, we reported the time of identifying the first BC (``$T_{FBC}$'') and the size of the best BC (the most compact BC ``$S_{BBC}$''). \LOGION not only identifies first BCs fast but also performs better search bias towards compact formulae.

We notice that none of the approaches identified BCs on the AMBA. We think that the reason is as follows. 
For the Tableaux-based approach, it cannot generate the tableau structure of AMBA within 3 hours.
Actually, it can only generate tableau structure of RP1, RP2, TCP, AAP, and ATM within 3 hours.
For the GA-based approach and \LOGION, to verify whether a formula is BC in AMBA,
they need to call the LTL satisfiability checker 24 times (logical inconsistency 1 time, minimality 21 times, and non-trivial 2 times). 
Each time the size of the formula checked by the LTL satisfiability checker is large.
Therefore, both approaches can only search for a small number of formulae within an hour on AMBA, resulting in not finding BCs.

In short, \LOGION significantly outperforms the state-of-the-art approaches, especially in computing BCs and general BCs.

\subsection{Application}
The direct application of BC is used to repair requirements specifications. 
One of the common strategies for resolving divergence is goal weakening. Its principle is to weaken the formulation of one or several among the divergent goals so as to make divergence disappear.
For Example \ref{examp:traincrossing}, Emmanuel Letier et al~\cite{letier2001reasoning} resolved divergence by weakening the first goal to cover $BC_2$ ($\future(h \wedge m)$). 
The first goal after the change is ``the pump is switched on when the water level is high and there is no methane''. Formally, $\globally((h \wedge \neg m) \rightarrow \nexttime(p))$.

BCs are also used to explain the synthesis unrealizability. 
For Example \ref{examp:traincrossing}, we ask some synthesis tools, like Ratsy \cite{bloem2010ratsy}, to build a controller that satisfies the specified goals. We will get as an answer that the specification is unrealizable. Recall that two BCs are presented in example 1. These formulae give us information about some admissible behaviors of the system, that lead us to violate the goals. It means that the environment always has a winning strategy to make the controller reach the BC. Such a BC could be thought of as an explanation of why the controller cannot satisfy all goals.

%% file: document/table3.tex
\begin{table}
    \centering
    \caption{Experimental results on the similarity between boundary conditions successively computed by \LOGION}
    \label{tab:rq1-2}
\begin{tabular}{c|c|c}
	\hline
	\multirow{2}[2]{*}{Case} & \multicolumn{2}{c}{LOGION} \bigstrut[t]\\
	& $\delta(BC_{i}, BC_{i+1})$ & $\Delta(BC_{i}, BC_{i+1})$ \bigstrut[b]\\
	\hline
	RP1   & 6.43  & 0.18  \bigstrut[t]\\
	RP2   & 5.70  & 0.21  \\
	Ele   & 4.72  & 0.25  \\
	TCP   & 5.24  & 0.23  \\
	AAP   & 5.32  & 0.31  \\
	MP    & 4.69  & 0.21  \\
	ATM   & 5.79  & 0.22  \\
	RRCS  & 5.70  & 0.19  \\
	Tel   & 5.39  & 0.24  \\
	LAS   & 5.75  & 0.09  \\
	PA    & 4.36  & 0.09  \\
	RRA   & 5.36  & 0.14  \\
	SA    & 5.70  & 0.10  \\
	LB    & 10.43  & 0.05  \\
	LC    & 17.72  & 0.10  \\
	AMBA  & 24.67  & 0.09  \bigstrut[b]\\
	\hline
\end{tabular}%
\normalsize
\begin{flushleft}
	\small
    ``$\delta(BC_{i}, BC_{i+1})$'' stands for the average distance between boundary condition $BC_i$ and its next boundary condition $BC_{i+1}$ successively computed.
    ``$\Delta(BC_{i}, BC_{i+1})$'' means the average normalized distance between boundary condition $BC_i$ and its next boundary condition $BC_{i+1}$ successively computed.
\end{flushleft}
\end{table}

%% file: document/table4.tex
\begin{table*}[t]
    \centering
    \caption{Experimental results of \LOGION, the GA-based approach and the Tableaux-based approach (Tableaux for short)}
\begin{tabular}{p{2.94em}|p{1.94em}p{1.94em}p{1.94em}p{1.94em}|p{1.94em}p{1.94em}p{1.94em}p{1.94em}p{1.94em}r|p{1.94em}p{1.94em}p{1.94em}p{1.94em}p{1.94em}r}
	\hline
	\multirow{2}[2]{*}{Case} & \multicolumn{3}{c}{Tableaux} & \multicolumn{1}{c|}{} & \multicolumn{6}{c|}{GA-based Approach} & \multicolumn{6}{c}{\LOGION} \bigstrut[t]\\
	\multicolumn{1}{c|}{} & \#BC & \#gen. & rank & Time & \#BC & \#gen. & rank & $T_{FBC}$ & $S_{BBC}$ & \multicolumn{1}{p{2.125em}|}{\#suc.} & \#BC & \#gen. & rank & $T_{FBC}$ & $S_{BBC}$ & \multicolumn{1}{p{2.125em}}{\#suc.} \bigstrut[b]\\
	\hline
	RP1 & \multicolumn{1}{r}{1.0 } & \multicolumn{1}{r}{0.9 } & \multicolumn{1}{r}{2.9 } & \multicolumn{1}{r|}{0.1 } & \multicolumn{1}{r}{45.2 } & \multicolumn{1}{r}{2.8 } & \multicolumn{1}{r}{\textbf{1.4 }} & \multicolumn{1}{r}{8.8 } & \multicolumn{1}{r}{14.2 } & \textbf{10} & \multicolumn{1}{r}{\textbf{6935.5 }} & \multicolumn{1}{r}{\textbf{4.0 }} & \multicolumn{1}{r}{1.5 } & \multicolumn{1}{r}{\textbf{3.5 }} & \multicolumn{1}{r}{\textbf{8.8 }} & \textbf{10 } \bigstrut[t]\\
	RP2 & \multicolumn{1}{r}{1.0 } & \multicolumn{1}{r}{0.7 } & \multicolumn{1}{r}{1.6 } & \multicolumn{1}{r|}{0.4 } & \multicolumn{1}{r}{48.0 } & \multicolumn{1}{r}{3.0 } & \multicolumn{1}{r}{1.4 } & \multicolumn{1}{r}{18.7 } & \multicolumn{1}{r}{16.8 } & \textbf{10} & \multicolumn{1}{r}{\textbf{14158.0 }} & \multicolumn{1}{r}{\textbf{7.0 }} & \multicolumn{1}{r}{\textbf{1.3 }} & \multicolumn{1}{r}{\textbf{1.4 }} & \multicolumn{1}{r}{\textbf{7.6 }} & \textbf{10 } \\
	Ele & \multicolumn{1}{r}{- }& \multicolumn{1}{r}{- }& \multicolumn{1}{r}{- }& \multicolumn{1}{r|}{- }& \multicolumn{1}{r}{41.2 } & \multicolumn{1}{r}{2.1 } & \multicolumn{1}{r}{\textbf{1.5 }} & \multicolumn{1}{r}{4.1 } & \multicolumn{1}{r}{\textbf{5.0 }} & \textbf{10} & \multicolumn{1}{r}{\textbf{7759.4 }} & \multicolumn{1}{r}{\textbf{22.2 }} & \multicolumn{1}{r}{\textbf{1.5 }} & \multicolumn{1}{r}{\textbf{1.8 }} & \multicolumn{1}{r}{\textbf{5.0 }} & \textbf{10 } \\
	TCP & \multicolumn{1}{r}{2.0 } & \multicolumn{1}{r}{0.6 } & \multicolumn{1}{r}{2.2 } & \multicolumn{1}{r|}{1.0 } & \multicolumn{1}{r}{80.4 } & \multicolumn{1}{r}{0.4 } & \multicolumn{1}{r}{2.0 } & \multicolumn{1}{r}{17.2 } & \multicolumn{1}{r}{\textbf{5.1 }} & \textbf{10} & \multicolumn{1}{r}{\textbf{10335.4 }} & \multicolumn{1}{r}{\textbf{20.0 }} & \multicolumn{1}{r}{\textbf{1.0 }} & \multicolumn{1}{r}{\textbf{0.9 }} & \multicolumn{1}{r}{7.3 } & \textbf{10 } \\
	AAP & \multicolumn{1}{r}{4.0 } & \multicolumn{1}{r}{4.0 } & \multicolumn{1}{r}{2.3 } & \multicolumn{1}{r|}{0.3 } & \multicolumn{1}{r}{56.8 } & \multicolumn{1}{r}{1.1 } & \multicolumn{1}{r}{2.2 } & \multicolumn{1}{r}{25.1 } & \multicolumn{1}{r}{\textbf{3.0 }} & \textbf{10} & \multicolumn{1}{r}{\textbf{22512.9 }} & \multicolumn{1}{r}{\textbf{41.9 }} & \multicolumn{1}{r}{\textbf{1.1 }} & \multicolumn{1}{r}{\textbf{4.4 }} & \multicolumn{1}{r}{\textbf{3.0 }} & \textbf{10 } \\
	MP & \multicolumn{1}{r}{- }& \multicolumn{1}{r}{- }& \multicolumn{1}{r}{- }& \multicolumn{1}{r|}{- }& \multicolumn{1}{r}{75.2 } & \multicolumn{1}{r}{3.0 } & \multicolumn{1}{r}{1.9 } & \multicolumn{1}{r}{10.0 } & \multicolumn{1}{r}{\textbf{3.0 }} & \textbf{10} & \multicolumn{1}{r}{\textbf{27112.2 }} & \multicolumn{1}{r}{\textbf{26.9 }} & \multicolumn{1}{r}{\textbf{1.1 }} & \multicolumn{1}{r}{\textbf{0.3 }} & \multicolumn{1}{r}{\textbf{3.0 }} & \textbf{10 } \\
	ATM & \multicolumn{1}{r}{3.0 } & \multicolumn{1}{r}{2.0 } & \multicolumn{1}{r}{2.4 } & \multicolumn{1}{r|}{1.9 } & \multicolumn{1}{r}{98.6 } & \multicolumn{1}{r}{1.4 } & \multicolumn{1}{r}{1.9 } & \multicolumn{1}{r}{12.0 } & \multicolumn{1}{r}{\textbf{6.0 }} & \textbf{10} & \multicolumn{1}{r}{\textbf{8773.6 }} & \multicolumn{1}{r}{\textbf{88.3 }} & \multicolumn{1}{r}{\textbf{1.4 }} & \multicolumn{1}{r}{\textbf{0.6 }} & \multicolumn{1}{r}{12.4 } & \textbf{10 } \\
	RRCS & \multicolumn{1}{r}{- }& \multicolumn{1}{r}{- }& \multicolumn{1}{r}{- }& \multicolumn{1}{r|}{- }& \multicolumn{1}{r}{60.4 } & \multicolumn{1}{r}{4.2 } & \multicolumn{1}{r}{1.6 } & \multicolumn{1}{r}{7.5 } & \multicolumn{1}{r}{11.0 } & \textbf{10} & \multicolumn{1}{r}{\textbf{17912.0 }} & \multicolumn{1}{r}{\textbf{113.4 }} & \multicolumn{1}{r}{\textbf{1.2 }} & \multicolumn{1}{r}{\textbf{1.1 }} & \multicolumn{1}{r}{\textbf{6.4 }} & \textbf{10 } \\
	Tel & \multicolumn{1}{r}{- }& \multicolumn{1}{r}{- }& \multicolumn{1}{r}{- }& \multicolumn{1}{r|}{- }& \multicolumn{1}{r}{155.0 } & \multicolumn{1}{r}{4.2 } & \multicolumn{1}{r}{2.0 } & \multicolumn{1}{r}{61.3 } & \multicolumn{1}{r}{12.7 } & 9 & \multicolumn{1}{r}{\textbf{12408.2 }} & \multicolumn{1}{r}{\textbf{89.1 }} & \multicolumn{1}{r}{\textbf{1.0 }} & \multicolumn{1}{r}{\textbf{1.0 }} & \multicolumn{1}{r}{\textbf{6.6 }} & \textbf{10 } \\
	LAS & \multicolumn{1}{r}{- }& \multicolumn{1}{r}{- }& \multicolumn{1}{r}{- }& \multicolumn{1}{r|}{- }& \multicolumn{1}{r}{36.8 } & \multicolumn{1}{r}{0.9 } & \multicolumn{1}{r}{\textbf{1.2 }} & \multicolumn{1}{r}{913.3 } & \multicolumn{1}{r}{43.0 } & 6 & \multicolumn{1}{r}{\textbf{12125.4 }} & \multicolumn{1}{r}{\textbf{133.3 }} & \multicolumn{1}{r}{1.3 } & \multicolumn{1}{r}{\textbf{1.2 }} & \multicolumn{1}{r}{\textbf{20.6 }} & \textbf{10 } \\
	PA & \multicolumn{1}{r}{- }& \multicolumn{1}{r}{- }& \multicolumn{1}{r}{- }& \multicolumn{1}{r|}{- }& \multicolumn{1}{r}{- }& \multicolumn{1}{r}{- }& \multicolumn{1}{r}{- }& \multicolumn{1}{r}{- }& \multicolumn{1}{r}{- }& 0 & \multicolumn{1}{r}{\textbf{3220.6 }} & \multicolumn{1}{r}{\textbf{45.9 }} & \multicolumn{1}{r}{\textbf{1.0 }} & \multicolumn{1}{r}{\textbf{2.4 }} & \multicolumn{1}{r}{\textbf{15.4 }} & \textbf{10 } \\
	RRA & \multicolumn{1}{r}{- }& \multicolumn{1}{r}{- }& \multicolumn{1}{r}{- }& \multicolumn{1}{r|}{- }& \multicolumn{1}{r}{88.9 } & \multicolumn{1}{r}{0.4 } & \multicolumn{1}{r}{2.0 } & \multicolumn{1}{r}{470.1 } & \multicolumn{1}{r}{12.6 } & \textbf{10} & \multicolumn{1}{r}{\textbf{1966.5 }} & \multicolumn{1}{r}{\textbf{80.2 }} & \multicolumn{1}{r}{\textbf{1.0 }} & \multicolumn{1}{r}{\textbf{1.2 }} & \multicolumn{1}{r}{\textbf{10.9 }} & \textbf{10 } \\
	SA & \multicolumn{1}{r}{- }& \multicolumn{1}{r}{- }& \multicolumn{1}{r}{- }& \multicolumn{1}{r|}{- }& \multicolumn{1}{r}{- }& \multicolumn{1}{r}{- }& \multicolumn{1}{r}{- }& \multicolumn{1}{r}{- }& \multicolumn{1}{r}{- }& 0 & \multicolumn{1}{r}{\textbf{3600.8 }} & \multicolumn{1}{r}{\textbf{42.9 }} & \multicolumn{1}{r}{\textbf{1.0 }} & \multicolumn{1}{r}{\textbf{11.3 }} & \multicolumn{1}{r}{\textbf{20.7 }} & \textbf{10 } \\
	LB & \multicolumn{1}{r}{- }& \multicolumn{1}{r}{- }& \multicolumn{1}{r}{- }& \multicolumn{1}{r|}{- }& \multicolumn{1}{r}{- }& \multicolumn{1}{r}{- }& \multicolumn{1}{r}{- }& \multicolumn{1}{r}{- }& \multicolumn{1}{r}{- }& 0 & \multicolumn{1}{r}{\textbf{87.8 }} & \multicolumn{1}{r}{\textbf{2.8 }} & \multicolumn{1}{r}{\textbf{1.0 }} & \multicolumn{1}{r}{\textbf{1489.2 }} & \multicolumn{1}{r}{\textbf{62.5 }} & \textbf{4 } \\
	LC & \multicolumn{1}{r}{- }& \multicolumn{1}{r}{- }& \multicolumn{1}{r}{- }& \multicolumn{1}{r|}{- }& \multicolumn{1}{r}{- }& \multicolumn{1}{r}{- }& \multicolumn{1}{r}{- }& \multicolumn{1}{r}{- }& \multicolumn{1}{r}{- }& 0 & \multicolumn{1}{r}{\textbf{16.6 }} & \multicolumn{1}{r}{\textbf{1.7 }} & \multicolumn{1}{r}{\textbf{1.0 }} & \multicolumn{1}{r}{\textbf{389.3 }} & \multicolumn{1}{r}{\textbf{91.6 }} & \textbf{7 } \\
	AMBA & \multicolumn{1}{r}{- }& \multicolumn{1}{r}{- }& \multicolumn{1}{r}{- }& \multicolumn{1}{r|}{- }& \multicolumn{1}{r}{- }& \multicolumn{1}{r}{- }& \multicolumn{1}{r}{- }& \multicolumn{1}{r}{- }& \multicolumn{1}{r}{- }& 0 & \multicolumn{1}{r}{- }& \multicolumn{1}{r}{- }& \multicolumn{1}{r}{- }& \multicolumn{1}{r}{- }& \multicolumn{1}{r}{- }& 0  \bigstrut[b]\\
	\hline
\end{tabular}%
\label{tab:rq2-1}%
\normalsize
\begin{flushleft}
	\small
	``\#BC'' stands for the average number of BCs in 10 runs. 
	The definitions of ``\#gen.'' and ``rank'' are shown in Section  \ref{sec:experimetalSettings}.
	``$T_{FBC}$'' means the time (second) of identifying the first BC. ``$S_{BBC}$'' means the size of the best BC (the most compact BC).
	``\#suc.'' is the number of successful runs (out of 10 runs).
	Finally, ``-'' means the failed case.
\end{flushleft}
\end{table*}%

%% file: document/Related_Wrok.tex
\section{Related Work}\label{sec:relatedwork}
Besides the inconsistency management approaches based on the informal or semi-formal way, such as \cite{hausmann2002detection,herzig2014conceptual,kamalrudin2009automated,kamalrudin2011improving},
a series of formal approaches \cite{ellen2014detecting,ernst2012agile,harel2005synthesis,nguyen2014kbre} recently have been proposed, which only focus on logical inconsistency or ontology mismatch.
Another related approach is Nuseibeh and Russo's work \cite{nuseibeh1999using}, which generates a conjunction of ground literals as an explanation for the unsatisfiable specification based on abduction reasoning.
For consistency checking methods, we have to mention the approach of Harel \emph{et al.} \cite{harel2005synthesis}, which identifies inconsistencies between two requirements represented as conditional scenarios.
While in this paper, we are interested in identifying the situations that lead to goal divergences, which are nothing but weak inconsistencies.
The works on reasoning about conflicts in requirements also include \cite{jureta2010techne, liu2010ontology,mairiza2010constructing}.

For the assessment of conflicts, Degiovanni \emph{et al.}~\cite{degiovanni2018goal} recently have proposed an automated approach to assess how likely conflict is, under an assumption that all events are equally likely.
For the resolution of conflicts, Murukannaiah \emph{et al.}~\cite{murukannaiah2015resolving} resolved the conflicts among stakeholder goals of system-to-be based on the Analysis of Competing Hypotheses technique and argumentation patterns. Related works on conflict resolution also include \cite{felfernig2009plausible} which calculates the personalized repairs for the conflicts of requirements with the principle of Model-based Diagnosis.
However, these approaches presuppose that the conflicts have been already identified and our approach for BC discovery provides a footstone for solving these problems.

In goal-oriented requirements engineering, we have to mention the work on obstacle analysis.
An obstacle, first proposed in \cite{van2000handling}, is a particular goal conflict, which captures the situation that only one goal is inconsistent with the domain properties.
Alrajeh \emph{et al.} \cite{alrajeh2012generating} exploited model checking technique to generate tracks that violate or satisfy the goals, and then compute obstacles from these tracks based on the machine learning technique. Other approaches of obstacle analysis include \cite{cailliau2012probabilistic,cailliau2014integrating,cailliau2015handling,van2000handling}.
Whereas, as obstacles only capture the inconsistency for single goals, these approaches fail to deal with the situations where multiple goals are conflicting.

Let us come back to the goal-conflict identification problem.
The concept of goal conflict was first proposed by van Lamsweerde \emph{et al.} \cite{van1998managing}, who also proposed a pattern-based approach to identify a goal conflict in a requirement specification. But the syntactical restrictions on the goal specifications and the ability of computing only one BC indeed limit the applicability of the approach.
While our approach \LOGION has no limitation on the specifications and is able to generate BCs in any form theoretically.

In 2016, Degiovanni \emph{et al.}~\cite{degiovanni2016goal} paid attention on goal-conflict identification again.
They provided a tableaux-based approach to generate BCs, consisting of two phases. It first generates a conjunction of domain properties and goals ($Dom \wedge G$) via a tableau, 
then identifies BCs with a complex logical algorithm based on tableaux.
However, for the specifications with a large number of domain properties and goals, the approach suffers from an efficiencies issue because tableaux are difficult to be generated.
It limits the approach to be applicable only on small specifications.
As shown in the last section, our approach \LOGION, as an anytime algorithm, generates significantly more BCs and solves more cases within the same time interval.

Another related work lies in Degiovanni \emph{et al.}\cite{degiovanni2018genetic} which applies a genetic algorithm to BCs. Based on the definition of BCs, they defined a fitness function to guide towards finding compact BCs.
However, our approach is more efficient than the GA-based approach and generates significantly more BCs.
Because the neighborhood relation of formulae captures the structural similarity, our approach finds another BCs within a few iterations once a BC is found.

%% file: document/Conclusion.tex
\section{Conclusion and Future Work}\label{sec:conclusion}
In this paper,
we discover a frequent phenomenon that some BCs are similar on the formula structure and give a formal analysis \wrt the formula distance.
Based on such an observation,
we present
an efficient local search algorithm \LOGION, to automatically identify BCs, which is featured as capturing the similarity in formula structure of BCs.
By taking experiments on the classical cases, we show that our approach \LOGION is more efficient than the state-of-the-art approaches of computing BCs and general BCs.
In future work, we hope to optimize the BC verification procedure by reducing the calls of the LTL satisfiability checker.